\documentclass[a4paper, 12pt]{article}
\usepackage{calrsfs}
\usepackage{amsmath,amssymb, amsthm, graphicx}
\usepackage{mathrsfs,enumerate,latexsym}
\usepackage[hidelinks]{hyperref,xcolor}
\usepackage{fullpage}
\usepackage{abstract}
\usepackage{bbm,soul}
\usepackage[comma]{natbib}

\usepackage[]{appendix}
\usepackage{doi}
\usepackage{enumitem,soul}
\usepackage{setspace}
\usepackage{float}

\expandafter\def\expandafter\normalsize\expandafter{%
	\normalsize
	\setlength\abovedisplayskip{5pt}
	\setlength\belowdisplayskip{5pt}
	\setlength\abovedisplayshortskip{5pt}
	\setlength\belowdisplayshortskip{5pt}
}

\graphicspath{{Figures/}}
\usepackage{subcaption}

\newtheoremstyle{break}
{\topsep}{\topsep}%
{\itshape}{}%
{\bfseries}{}%
{\newline}{}%
\theoremstyle{break}
\newtheorem{proposition}{Proposition}
\newtheorem{assumption} {Assumption}

\newtheorem{lemma}{Lemma}
\newtheorem{theorem}{Theorem}

\newtheorem{innercustomthm}{Axiom}
\newenvironment{customax}[1]
{\renewcommand\theinnercustomthm{#1}\innercustomthm}
{\endinnercustomthm}

\renewcommand{\vec}[1]{\mathbf{#1}}
\renewcommand{\vec}[1]{\boldsymbol{#1}}
\newcommand{\Real}{\mathbb{R}}

\newcommand{\Realo}{\mathbb{R}_{\geq 0}}
\newcommand{\upalpha}{\bar{\alpha}}
\newcommand{\loalpha}{\underline{\alpha}}

\newcommand{\network}{(\mathbf{g}, \vec{\pi})}
\newcommand{\networkp}{(\mathbf{g}',\vec{\pi}')}
\newcommand{\graph}{\mathbf{g}}

\newcommand{\mass}{\pi}
\newcommand{\vecmass}{\vec{\mass}}
\newcommand{\Pol}{\mathcal{P}}

\newcommand{\Nspace}{\mathcal{N}}
\newcommand{\distance}{d}
\newcommand{\x}{\vec{x}}
\newcommand{\y}{\vec{y}}

\newcommand{\nodes}{N}
\newcommand{\norm}[1]{\left\lVert#1\right\rVert}

\usepackage{subcaption,booktabs}
\title{Polarization in Networks:\\ Identification-alienation Framework\thanks{ We thank Joan Esteban, Jan Fa\l{}kowski, Antonin Mac\'{e}, and Fernando Vega-Redondo for their comments and to Debraj Ray for our discussion at the  beginning of the project. This work was supported by the ANR [project ANR 18-CE26-0020-01].}}
\author{
Kenan Huremovi\'{c}\thanks{IMT School for Advanced Studies Lucca, Piazza S. Francesco, 19, 55100 Lucca, Italy. E-mail: \texttt{kenan.huremovic@imtlucca.it}.}
\and
Ali Ozkes\thanks{Wirtschaftsuniversit\"{a}t Wien, Institute for Markets and Strategy, Welthandelsplatz 1, 1020, Vienna, Austria. E-mail: \texttt{ali.ozkes@wu.ac.at}}
}
\usepackage{soul}

\begin{document}
\maketitle

\begin{abstract}

\begin{singlespace}
	We introduce a model of polarization in networks as a unifying setting for the measurement of polarization that covers a wide range of applications. We consider a substantially general setup for this purpose: node- and edge-weighted, undirected, and connected networks.  We generalize the axiomatic characterization of \cite{esteban1994} and show that only a particular instance within this class can be used justifiably to measure polarization in networks.   
\end{singlespace}
\end{abstract}

\textbf{JEL codes:}  D63, D70, P16 

\textbf{Keywords:} measurement, networks, polarization 
\newpage

\vspace*{-5pt}
\section{Introduction}


 Polarization in a population denotes an intensified disconnect among its groups. The analysis of the sources and the consequences of polarization depends highly on what is measured and how, which, in turn, is strictly contingent on the particular context. For instance, while in the context of American politics polarization is perceived as the division of masses into the cultural camps of liberals and conservatives, in the context of European multi-party parliaments, it is seen as the existence of ideologically cohesive and distinct party blocks.\footnote{See \cite{fiorina2005culture} and \cite{maoz2010political} for the two different contexts.} So even the term ``political polarization" is not indicative of what is being measured and how. Existing literature reflects this complexity, and there is an abundance of measures without a unified formalism that applies to comparable contexts. 
 

Although there are substantial differences among existing measures across fields, one ubiquitous feature can be identified. Namely, most of the current measures are proposed in settings with a uni-dimensional scalar attribute on which the polarization is assumed to occur. However, conflicts in societies are in general related to an irreducibly complex set of attributes and most of the empirical work rely on categorical data on various characteristics.\footnote{Examples include ethnolinguistics \citep{montalvo2008discrete}, ethnic power relations \citep{wimmer2009ethnic}, and political retweets \citep{conover2011political}.} Dimensionality reduction approaches are called for in many instances, because the existing polarization measures allow for only a uni-dimensional, or at most a bi-dimensional domain \citep[][]{hill2015disconnect}. However, reduced dimensions can be questionable for their capacity to represent the actual phenomenon of interest \citep{kam2017polarization}.

In this paper, we propose the formalism of network theory to study  the measurement of polarization as it delivers the desired generality and spans a large variety of contexts. We fully characterize a polarization measure {following the} axiomatic setting introduced {by \cite{esteban1994} (henceforth ER)} for distributions on the real line. Same as ER, we restrict ourselves to distributions with finite support.


Our setup is built on undirected networks in which both nodes and links are weighted. A node in the network represents a certain attribute or grouping of individuals in the population. The weight of a node corresponds to the number of individuals in the population that are characterized by the attribute or members of the group (\emph{e.g.,} a political party). The {weighted} links describe (direct) bilateral relationships between nodes. { This setup is quite general and can represent a wide range of settings in which measuring polarization is an issue of first-order importance. We describe a number of important examples in the next section, {with a particular focus on the political domain, not only because it is a central point of discussion, but also because it comprises of a variety of aspects that can be captured distinctly within network formalism}. For instance, we show how elite polarization can be modeled within our framework through networks of politicians, parties, or policy space. Mass polarization, on the other hand, can be modeled through a network of opinions or preferences. {Going beyond the political domain, we furthermore} discuss how our approach can be used to study polarization in any setting with multidimensional distributions with finite support. }

{The axiomatic approach developed by ER for distributions with finite support on the real line led to the development of measures in several other domains, such as measures for continuous distributions as in \cite{duclos2004polarization} and measures for binary classifications as in \cite{montalvo2008discrete} (henceforth MRQ). Most of the applications employing measures within this line of work lie in the fields of income inequality and social conflicts. In their seminal contribution in this context, ER conceptualize polarization as the aggregate antagonism in a population. The effective antagonism an individual feels against another depends on how alienated she feels from the other's group and how identified she feels with her own group. According to ER, a population in which individuals are identified within groups is polarized if there is a high level of intra-group homogeneity, a high level of inter-group heterogeneity, and a small number of large-enough groups. They deliver a characterization of a class of polarization measures, based on an axiomatization built around distributional properties and not confined to incomes or wealth, although the main motivations of ER were about income and wealth distributions.} 

{Following ER, we provide an axiomatic characterization for measures of network polarization. We argue that networks represent a powerful tool to capture any distribution with a finite support and a notion of distance.} Thus, the strength of our contribution lies in the fact that we deliver an axiomatic foundation for a family of measures that are applicable in a significantly  larger set of domains. Furthermore, as any distribution considered in ER or MRQ can be represented as a network, our work can be seen as a unifying generalization, with ER and MRQ as special cases.

The class of measures characterized by ER is identified by the range of values that parameter $\alpha$, which captures the importance of identification in the effective antagonism, can take. Our {first} result shows, quite surprisingly, that this class is thinned down by a unique value, \emph{i.e.,} $\alpha=1$ (Theorem \ref{thm:ThmMain}).\footnote{MRQ also identify $\alpha=1$ in their setup, which is a special case of ours. Furthermore, their axiomatization is different than ours and ER.} {Note that adaptations of these axioms are neither trivial nor straightforward, as networks allow for a much larger generality in representing discrete distributions than the real line. Recent literature on the measurement of polarization carried along the restricting assumption that the attributes can be captured by the values of a scalar variable. We take off where ER leave, and deliver an analysis that does not ``sweep a serious dimensionality issue under the rug" (ER, p. 823). Our approach accommodates a significantly larger variety of settings that are not confined to scalar attributes, and naturally include the case of the Euclidean distance on the real line as a special case.}  This entails a solution to an unresolved issue in this line of research as a by-product, in that our results point to the choice of an exact value within the interval $(0,\alpha^*\simeq 1.6]$.\footnote{ER proposes {further} restrictions in that regard by imposing an additional axiom (Axiom 4) {that brings about a lower bound, \emph{i.e.,} $\alpha\geq1$.}}

{It is desirable that polarization measures attain their maximum at the symmetric bipolar distribution.} Contrary to the real {intervals}, in networks there can be any finite number of nodes with {maximal} distance between them. Still, we show that any measure within the family we characterize is maximized at the symmetric bipolar distribution --- when the population is symmetrically distributed among the \textit{two} most distant nodes in the network (Proposition \ref{prop:MaxElement}).

Finally, we show that if we restrict our attention to particular classes of networks emerging in certain domains such as language trees ({class of tree networks}) or income distributions ({class of line networks}), one of the axioms, Axiom \ref{ax:Axiom3}, can be weakened in a systematic way to allow for a wider class of measures that can be used consistently (Theorem \ref{thm:ThmGen}).  For instance, in the special case of line networks that can be used to represent income distributions, our set of axioms and the class of measures reduce to the ones in ER.

\medskip

\begin{flushleft}{\textbf{\large{Related literature}}}\end{flushleft}
{It presents a challenge to pay a fair tribute to the ever-growing literature on the measurement of polarization. Here, we refer to a set of papers in different domains and discuss a few closely related ones. We mention several other works in} {Section \ref{sec:Conclusion}.}

{Polarization is studied in social sciences (particularly in economics and political science) in relation to economic inequality \citep{esteban2007extension,esteban2012comparing,zhang2001difference}, social conflict \citep{desmet2017culture,montalvo2008discrete,ostby2008polarization}, political economy \citep{aghion2004endogenous,desmet2012political,lindqvist2010political}, international relations
\citep{maoz2006network}, political ideologies
\citep{abramowitz2008polarization,fiorina2008political,lelkes2016mass,martin2017bias},  political sentiments 
\citep{boxell2017greater,garcia2015ideological}, and social attitudes
\citep{dimaggio1996have,lee2014social,mccright2011politicization},
among others. }

{We want to emphasize that we are not the first to consider an ER-type approach to the measurement of polarization in networks. For instance, both \cite{esteban1999conflict} and \cite{esteban2011linking} explore this issue. However, to the best of our knowledge, this is the first paper to provide an axiomatic characterization for measures of polarization in networks.\footnote{{\cite{esteban1999conflict} arrive at $\alpha=1$ in their attempt to connect the intensity of conflict to polarization, without an axiomatic discussion, {while \cite{esteban2011linking} supplement the four axioms in \cite{duclos2004polarization} with a fifth axiom that delivers $\alpha=1$.}}}
\cite{fowler2006connecting, fowler2006legislative} and \cite{maoz2006network} are among the leading examples where network formalism is proposed for the measurement of polarization, without an axiomatic treatment.\footnote{
The measure \cite{maoz2006network} uses is developed in the unpublished working paper by \cite{maoz2006WPnetwork}, and while inspired by \cite{duclos2004polarization}, it is only shown to satisfy an extended and qualitatively different set of properties.} Finally, \cite{permanyer2015measuring} characterize a distinct family of measures for categorical attributes by using identification-alienation framework and a number of additional axioms. }

The ER polarization index is often used in applied work, in some cases (for instance \cite{aghion2004endogenous, alesina2003fractionalization,collier2004greed, desmet2009linguistic} and \cite{dower2017colonial})  with data that cannot be represented as a distribution on the real line, but can be represented as a network. Our paper provides a justification to use the ER measure with $\alpha =1$ in such cases, as done in \cite{desmet2009linguistic} and \cite{dower2017colonial}, but also suggests that using values of $\alpha$ that are different than 1 (as in \cite{alesina2003fractionalization, aghion2004endogenous}, and \cite{collier2004greed}) may not be appropriate.

{The rest of the paper is  organized as follows. In Section \ref{sec:PolMeasure} we describe the environment we study and illustrate the {wide} applicability of our approach. In Section \ref{sec:Charachterization} we define polarization, state the axioms, and deliver our major results. In Section \ref{sec:Discussion} we discuss the importance of network structure in terms of polarization and formally illustrate the connection between our work and previous literature. We conclude in Section \ref{sec:Conclusion}. }

\vspace*{-5pt}
\section{Networks and polarization} \label{sec:PolMeasure}

We consider a population in which individuals belong to $n>0$ mutually exclusive groups of potentially different sizes. A group may be,  for instance,  a political party, ethnic group, or a set of individuals that share the same attributes. For each group  $i$, $\mass_i \geq 0$ denotes the number of individuals in group $i$. 
 When $\pi_i=0$ we say that group $i$ is \textit{empty}. Vector $\vecmass \in \Realo^{n}$ describes the distribution of a population among $n$ groups.

 Bilateral relationships between  groups are described with an undirected weighted  graph (UWG) $\graph$,  with the set of nodes, $N(\graph)$, equal to the set of groups, and the set of undirected links (or edges) $E(\graph) =\{\{i,j\}: \{i,j\} \in N^2 \text{ and } i \neq j\}$.  As usual, we denote the edge between nodes $i$ and $j$ in graph $\graph$ with  $ij$ and the weight of that edge with $g_{ij}\geq 0$.  While we treat weights quite generally, it is useful to think of  $g_{ij}$ as the direct distance between two connected nodes $i$ and $j$ --  a higher $g_{ij}$ implies a \textit{weaker} connection between $i$ and $j$.\footnote{The particular interpretation of weights $(g_{ij})_{i, j \in N}$ depends on the application, as we demonstrate in Section \ref{Sec:Examples}.} For the remaining part of the paper we write  $ij \in \graph$ instead of $\{i,j\} \in E(\graph)$ to indicate that there is an edge between nodes $i$ and $j$ in $\graph$. When nodes $i$ and $j$ are not directly connected, we write $ij \notin \graph$.  Moreover, since groups are represented as nodes,  we use words \textit{group} and \textit{node} interchangeably. 

We restrict our attention to connected graphs, \emph{i.e.}, graphs in which there is a path connecting any two nodes.\footnote{{We consider only connected graphs in this paper.} Our insights can be extended, in a somewhat ad-hoc manner, to cases when $\graph$ is unconnected, for instance, by defining the distance between nodes from different components of $\graph$ to be equal to the longest path between any two connected nodes in $\graph$.}  The distance between nodes $i$ and $j$ in $\graph$,  denoted with $d_\graph(i,j)$,  is measured using the notion of the shortest path. That is, while there may be different routes one can take to reach node $j$  starting from node $i$ and moving along the links in  $\graph$, the distance between $i$ and $j$ is the length of the shortest path. This notion of distance, also known as the \emph{geodesic distance}, is the standard in graph theory and the theory of networks \citep{newman2003structure,jackson2008}. 

Let $\mathcal{G}_n$ denote the set of all UWGs with $n$ nodes, and let $\{\mathcal{G}_n\}_{n \in \mathbb{N}}$  denote the family of all UWGs with any finite number of nodes. The main object of our analysis is the ordered pair $(\graph, \vecmass) \in \mathcal{G}_n \times \Realo^{n}$, which represents a weighted (node-weighted and link-weighted) network.  We use $\mathcal{N}$ to denote the set of all networks with finite number of nodes. 

In the special case when $\vecmass =\vec{1}$,  $\network$ coincides with the standard notion of a ({link}-) weighted network.\footnote{Alternatively, one can think of $\network$ as a distribution $\vecmass$ on graph $\graph$.}  If, additionally, $g_{ij}=1$ whenever $ij \in \graph$, then $\network$ is a binary network.  Thus $\network$ is a fairly general object that can be used to represent any undirected network we observe, allowing for {weights on nodes and edges.} In Section \ref{sec:Discussion}  we show that any distribution studied in ER or any {classification} {covered by} MRQ can be represented as a network.

A polarization measure is a mapping $\mathcal{P}: \Nspace \rightarrow \Realo$ that assigns to each network $\network \in \Nspace$  a non-negative real number. 

{Before turning to the axiomatic analysis, we discuss a number of examples in which data can be represented as a network and measuring polarization is of interest.} 
 
\subsection{Examples} \label{Sec:Examples}

\subsubsection{Polarization in political networks}
We consider several networks that arise in politics,  each of which encodes a different aspect of the prevailing political climate. In particular, we consider situations in which collection of individuals express their preferences over alternatives, natural examples of which include a parliament voting on bills and an electorate choosing among candidates. We discuss how these two can be modeled as networks in order to measure elite and mass polarization.\footnote{See \cite{kearney2019analyzing} for a review focusing on networks in the political domain from a general perspective.} 

We start with the case of a parliament with possibly more than two parties. Let there be $N\in\mathbb{N}$ representatives denoted by $\mathcal{R} = \{1,..., N\}$ and $T\in\mathbb{N}$ parties denoted by $\mathcal{T}=\{t_1,\dots,t_T\}$. Suppose there are $k\in\mathbb{N}$ bills that are sponsored by representatives, either individually or in groups, which are thereafter voted for approval in the parliament. Let $v_{ij}\in\{0,1\}$ denote the vote of $i$ for the bill $j\in \{1,2,...,k\}$ and $\mathcal{V}=\{0,1\}^k$ denote the set of possible vote combinations.

\textbf{Network of representatives, $(\graph', \vecmass')$, link-weighted.}\\
 The set of nodes in graph $\graph'$ is $\mathcal{R}=\{1,\dots,N\}$. For any two representatives $i$ and $j$, let $g'_{ij}\geq 0$ denote the share of bills on which they do not vote in the same way.\footnote{Alternatively, one can model that two representatives are connected (with weight $1$) if they vote together for more than $50\%$ of the bills and not connected otherwise, in which case we would have an unweighted network.} Thus, $g'_{ij}$ stands for the (inverse of the) strength of their connection, where $g'_{ij} =0$ indicates that $i$ and $j$ always vote the same way.\footnote{The fact that $g'_{ij} = 0$ does not indicate that link between $i$ and $j$ does not exist, but that the distance between $i$ and $j$ is 0.} When they never vote the same way on any bill, they are not directly connected, hence  $ij \notin \graph'$. The size of every node $i\in \mathcal{R}$ is $\pi'_i =1$, thus $\vecmass'={\bf 1}$, as each node represents a unique representative. An example of networks as such can be found in \cite{andris2015rise}.
 
 \textbf{Network of co-sponsorships, $(\hat{{\graph}}', \hat{\vecmass}')$, unweighted.}\\
 The set of nodes in $\hat{\graph}'$ is $\mathcal{R}=\{1,\dots,N\}$. $\hat{g}'_{ij} =1$ if $i$ and $j$ co-sponsored al least one bill together, and $ij \notin \hat{\graph}' $ otherwise.\footnote{ Alternatively,  $\hat{g}'_{ij}$  may reflect how many bills  $i$ and $j$ co-sponsored together, in which case, we would have a link-weighted network.} The size of each node  $i\in \mathcal{R}$ is $\hat{\pi}'_i =1$, thus $\hat{\vecmass}'={\bf 1}$, since each node represents a unique representative. \cite{fowler2006connecting} studies this type of networks.

\textbf{Network of votes, $(\tilde{\graph}', \tilde{\vecmass}')$, node-weighted.}\\
 The set of nodes in graph  $\tilde{\graph}'$ is $\mathcal{V}=\{v_1,\dots,v_{2^k}\}$.  Two nodes  (vote combinations) $v_i$ and $v_{j}$  are connected, \emph{i.e.}, $ij \in \tilde{\graph}'$, whenever  $v_{i}$ and $v_{j}$ differ only in a single coordinate (bill). Each link in $\tilde{g}'$ has a weight 1.
$\tilde{\pi}'_{{i}}$ denotes the number of individuals with voting profile $v_{i}$, and $\tilde{\vecmass}'$ is the corresponding distribution. \cite{brams2007minimax} and \cite{moody2013portrait}, among others, study this type of networks.

\textbf{Network of parties, $({\bf {\bar g}}',\bar{\vecmass}')$, node- and link-weighted.}\\
The set of nodes in ${\bf {\bar g}}'$ is $\mathcal{T}=\{t_1,\dots,t_T\}$. $\bar{g}'_{ij}$ denotes the share of bills on which a majority of representatives in both parties vote the same way.\footnote{$\bar{g}'_{ij}$ captures the ideological distance \emph{i.e.,} the extent the policies of two parties overlap, which can be measured in different ways. \cite{maoz2010political} take, for instance, the similarities in party manifestos.} Thus, $ij \notin{\bf {\bar g}}'$ indicates that there is no bill that is supported (or opposed) by a majority of representatives in both parties. The size of a node $t_i\in \mathcal{T}$, $\bar{\pi}'_i$, denotes the number of seats of the party $i$ in the parliament.\footnote{Alternatively, $\bar{\vecmass}'$ can be taken as ${\bf 1}$, disregarding party sizes and focusing on closeness among parties, in which case we would have a link-weighted network.} See \cite{maoz2010political} for an analysis on party networks.

Each network we describe above focuses on a different aspect of the political activities in the parliament. Accordingly, the corresponding measures of polarization provide different, yet complementary, insights into congressional polarization. For instance, $\mathcal{P}(\graph', \vecmass')$ tells us how polarized the policy positions of representatives based on their vote histories are, regardless of their party affiliations, whereas $\mathcal{P}(\bf {\bar g',\bar{\vecmass}'})$ measures the party-level polarization.\footnote{We write $\mathcal{P}\network$ in place of $\mathcal{P}(\network)$ with a slight abuse of notation.} Also,  $\mathcal{P}(\tilde{\graph}', \tilde{\vecmass}')$ is informative about the polarization with respect to policy space, while $\mathcal{P}(\hat{\graph}', \hat{\vecmass}')$ captures the polarization among representatives with respect to policy cooperation.   

For illustration, let us more closely compare networks  $(\graph', \vecmass')$ and $(\tilde{\graph}', \tilde{\vecmass}')$, which are based on exactly the same data, \emph{i.e.,} votes on bills.  Consider the following example with 3 bills and 8 representatives, where ``$+$" in \eqref{mat:NetParties} represents approval for a bill and ``$-$" represents disapproval. 

\begin{align}\label{mat:NetParties}
\begin{array}{rccccc}
 & R_{1-3} & R_4 & R_{5-6} & R_7 & R_8  \\ \cline{2-6}
\text{ I} & + & - & - &+ & + \\
\text{ II} & - & + & + &- & + \\
\text{ III}& - & - & + &+ & + \\
\end{array} 
\end{align}

{Panel (a) of Figure \ref{strategic} below shows the corresponding network of representatives, whereas the panel (b) shows the corresponding network of votes}.\footnote{Both networks  $(\graph', \vecmass')$ and $(\tilde{\graph}', \tilde{\vecmass}')$ have a level of ``structural regularity." Graph $\graph'$ leads to a complete network structure in the sense that each node is connected to any other node, even though there is a substantial heterogeneity across weights of the links. Graph $\tilde{\graph}'$ has a lattice structure. This is by no means necessary for our approach, which is applicable to connected networks with arbitrary structure.  For instance, as in \cite{andris2015rise}, two representatives can be connected if they vote the same way sufficiently many times, then the $\graph'$ will not have the complete graph structure. Co-sponsorship networks such as $(\hat{\graph}', \hat{\vecmass}')$ have, in general, quite irregular structures, as in \cite{fowler2006legislative}.} Since the two networks describe two different sets of relations in the legislation, we may expect that the measured level of polarization differs between them. Nevertheless, any polarization measure in our framework is applicable to both cases. To obtain a deeper insight, for instance, one can also compare polarization of networks representing different types of relationships with a suitable normalization \emph{e.g.,} by dividing the polarization index with the maximal value it can attain.

     \begin{figure}[h]
       \centering
       \subcaptionbox{Network of represe\-ntatives. Nodes denote representatives and two nodes are not connected if they do not agree on any issue. The thickness of edges indicate weights.\label{fig:Reps}}
          {\includegraphics[scale=1.5]{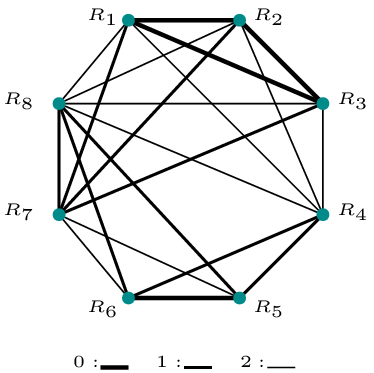}}
         \;\;\;\;\;        \;\;\;\;\;        \;\;\;\;\;
       \subcaptionbox{Network of votes. The nodes represent all possible vote combinations \emph{e.g.}, 100 represents the approval of only first bill.\label{fig:Votes}}
           {\includegraphics[scale=1.25]{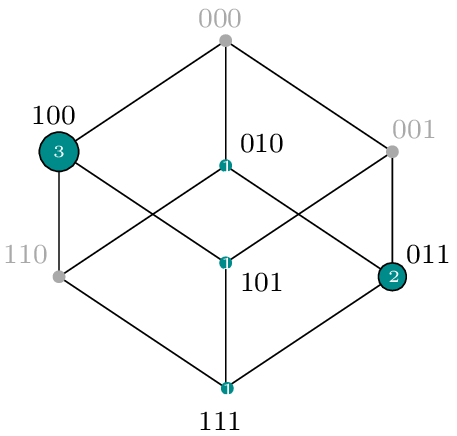}}
       \caption{\small Two possible network representations of the same profile of votes of representatives.}\label{strategic}
     \end{figure}

We next turn to the case of mass polarization. Our example is concerned with an electorate choosing among candidates for an office (or individuals expressing preferences over {policy} alternatives such as \emph{remain}, \emph{soft-Brexit}, and \emph{hard-Brexit}).\footnote{Many other networks can be considered in the context of voter preferences and affective (mass) polarization has been a major concern in recent years. For instance, for the context of European multi-party systems, \cite{reiljan2020fear} proposes a measure based on the divergence of partisan affective evaluations between in-party and out-parties, which could be represented on a network. Here, we take the example of ordinal preferences for the sake of the simplicity of the exposition.} Let there be a set of alternatives $X=\{x_1,\dots,x_m\}$ and each individual $i\in\{1,\dots,n\}$ be endowed with a preference $P_i\subseteq X\times X$ that is a linear order, \emph{i.e.}, a complete, antisymmetric, and transitive binary relation on $X.$ Let $\mathcal{L}$ denote the set of all preferences over $X$ and $\mathcal{L}^n$ be the set of profile of preferences.

\textbf{Network of preferences, $(\graph'', \vecmass'')$.} The set of nodes is $\mathcal{L}=(p_1,\dots,p_{m!})$. Two nodes $p_i$ and $p_j$ are connected with  $g''_{ij} =1$, whenever $p_i$ can be obtained from $p_j$ by switching only one binary preference, \emph{i.e.,} the Kemeny distance between $p_i$ and $p_j$ is 1 \citep{kemeny1959mathematics}.\footnote{A network of preferences can be represented as a special network of votes, in which each bill represents a pairwise comparison of alternatives and transitivity is imposed.} We denote with ${\pi}''_{{i}}$ the  number of individuals with preference $p_{i}$, and with ${\vecmass}''$ the corresponding distribution. See \cite{cervone2012voting} for a study on preference networks.\footnote{Often without explicitly using the language of networks, graph theoretical representations of preferences are studied in the social choice literature widely. There is also a growing interest in measuring polarization in preference profiles, as in \cite{can2015measuring,can2017generalized}. Note that network $(\graph'', \vecmass'')$ could  alternatively be defined using a weighted metric as in \cite{can2014weighted}.}

For an illustration, let $\{a,b,c\}$ be the set of alternatives and consider the preference profile with 11 (millions of) individuals represented by \eqref{mat:Preferences}. 

\begin{align}\label{mat:Preferences}
\begin{array}{cccc}
2 & 3 & 2 & 4  \\ \toprule
a & b & c & c \\ 
b & a & a & b\\ 
c & c & b & a \\ 
\end{array} 
\end{align}

This profile of preferences can be represented with a network as depicted as in Figure \ref{fig:pref}.

\begin{figure}[h]
      \centering
        {\includegraphics[scale=1.25]{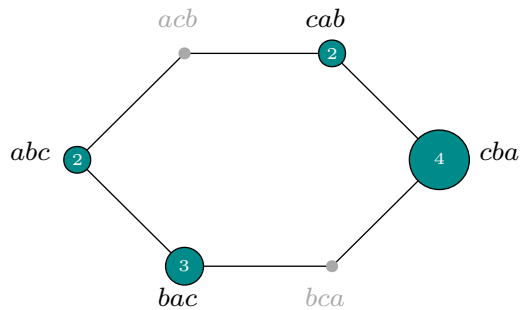}}
          \vspace{-2mm}
      \caption{\small A distribution over a preference network with 3 alternatives and 11 individuals.}
      \label{fig:pref}
    \end{figure}

\subsubsection{Beyond the political domain}

{While we paid a close attention to examples of networks from the political domain,} our approach can naturally be applied in a much wider range of applications, {not necessarily confined to those that are commonly studied using networks}. For instance, our setting can be adopted to study multidimensional polarization in any distribution with a discrete support. To see how, take the example of polarization in a society with respect to income and education (both measured on some discrete, increasing scale).  The set of all pairs of  income ($\iota$) and education ($\epsilon$) levels defines the set of nodes in the network. Two nodes $\x = (x_\iota, x_\epsilon)$ and $\y = (y_\iota, y_\epsilon)$ are connected, with link $\x\y$  of weight $s$ ($g_{\x\y} =s$),  if, for instance,  $|x_\iota - y_\iota| + |x_\epsilon-y_\epsilon|=s$, that is if the Manhattan distance between $\x$ and $\y$ is equal to $s$. 


Other potential applications include conflicts between groups \citep{esteban1999conflict,esteban2011linking}, private provision of public goods \citep{bramoulle2007public},  research output and citation networks \citep{leskovec2005graphs}, friendship networks \citep{calvo2009peer}, and trust networks \citep{richardson2003trust}.

\section{Identification-alienation framework and axiomatization} \label{sec:Charachterization}

To recall, our objective in this paper is two-fold. First, 
we propose network theory as a unifying {formalism} to study polarization without any constraint on dimensionality. Second, we present a theoretical foundation for a family of polarization measures in this {setting}. For the latter, we closely follow the axiomatic approach in ER, who envisage polarization as the aggregate antagonism in a population, {based on the identification and alienation among individuals}. 

First, as in ER, we require polarization measures to satisfy the following property that ensures invariance of the measure  with respect to the size of the population $\sum_{i \in \nodes(\graph)}\pi_i$. {Thus, in fact, ${\vecmass}$ may represent also a probability mass function.}

\begin{assumption}[Homotheticity]\label{ass:homothetic}
	$\mathcal{P}\network \geq \mathcal{P}\networkp$ $\implies$ $\mathcal{P}(\graph, \lambda \vecmass) \geq \mathcal{P}(\graph', \lambda \vecmass')$ for all $\network,\networkp$ $\in \Nspace$ and  $\lambda >0$.
	
\end{assumption}	

The antagonism between individuals depend on how they identify themselves and how alienated they feel from others.  In the network setup we propose, individuals in a population are identified only with their definitive attributes, which are represented as nodes in the network.  As emphasized before, these attributes are by no means restricted to singletons or a uni-dimensional space. 

The effect of the feeling of \emph{identification} of each individual on her antagonism towards another is measured in relation to the presence of others that share the same attributes, hence are in the same node. This effect is the basis of the intra-group homogeneity, and we denote it with $I(\pi_i)$. Thus, when the nodes represent individuals, each individual feels the same level of identification, whereas when nodes represent groups of individuals, the identification an individual feels is a function of the size its node ($I(\pi_i)$).\footnote{This implies that two groups (nodes) of the same size exhibit the same level of identification. While potentially restrictive, this is standard in the identification-alienation framework \citep{esteban1994,esteban2012comparing}.}  The only assumption we make on the identification function $I:\Realo \rightarrow \Realo$ is that $I(\pi_i)>0$ whenever $\pi_{i}>0$. 

The distance an individual perceives between herself and any other individual is a natural component of the {antagonism} between individuals as it forms the basis of the inter-group heterogeneity. We measure this \emph{alienation} component as a function of the distance between individuals $a(d(i,j))$. We assume that the {alienation function} $a: \Realo \rightarrow \Realo$  is a continuous and nondecreasing  function  with $a(0)=0$.   

Finally, the \textit{effective antagonism }of group $i$ towards group $j$ is measured by continuous  and strictly increasing function $T(I_i,a_{ij})$ of the identification of group $i$, $I_i=I(\pi_i)$, and the alienation between groups $i$ and $j$, $a_{ij}= a(d(i,j))$, satisfying $T(I_i, 0) = 0$. As in ER, we consider polarization measures $\mathcal{P}: \Nspace \rightarrow \Realo$  defined as the sum of effective antagonisms:
\begin{align}\label{eq:PolarizationGeneral}
	\mathcal{P}\network = \sum_{i=1}^{n}\sum_{j=1}^{n}\pi_i \pi_{j}T \Big(I(\pi_i), a\big(d_\graph(i,j)\big) \Big).
\end{align}
As we shall see, our axioms will pin down specific functional form for $T \Big(I(\pi_i), a\big(d_\graph(i,j)\big) \Big)$.

Our goal is to follow the axiomatization in ER as closely as possible, and modify it only when the network setting requires.  As it turns out, the first two axioms can be restated only with slight changes in the nomenclature. Axiom 3 needs an important adjustment.

 \begin{customax}{1}\label{ax:Axiom1}
	\textit{Data:} Network $\network$ with $n \geq 3$ nodes such that $\pi_x>\pi_y=\pi_z>0$ and $\pi_i=0$  $\forall i \in \nodes(\graph)\setminus\{x,y,z\}$.  Furthermore, $d_{\graph}(x,y) \leq d_{\graph}(x,z)$. \\
	\textit{Statement:} Fix $\pi_x$ and $d_{\graph}(x,y)$. There exists $\epsilon> 0$ and $\mu=\mu(\pi_x, d_{\graph}(x,y)) >0$ such that $d_{\graph}(y,z)< \epsilon$ and $\pi_y< \mu \pi_x$ imply that for any $\networkp \in \Nspace$ with $n\geq 2$ nodes such that $\mass'_{x'} =\mass_{x}$, $\mass'_{w'} = \mass_y + \mass_z$,  $d_{\graph'}(x', w') = \frac{1}{2}\left( d_{\graph}(x,y) + d_{\graph}(x,z)\right)$ and $\mass'_{i'}=0, \;i' \in \nodes(\graph') \setminus \{x', w'\}$, we have  $\mathcal{P}\networkp > \mathcal{P}\network$.
	
\end{customax}

The Axiom \ref{ax:Axiom1} captures the situations where two small groups join while keeping the (average) distance the same. 

\begin{figure}[H]
	\centering
	\subcaptionbox{ The move shown by arrows increase polarization.\label{fig:ax1}}
	{\includegraphics[scale=1.75]{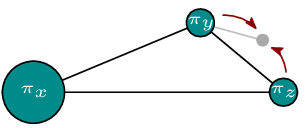}}
	\;\;\;\;\;        \;\;\;\;\;        \;\;\;\;\;
	\subcaptionbox{Axiom 1 applies when the new node is further away from the two small nodes as well.\label{fig:ax11}}
	{\includegraphics[scale=1.5]{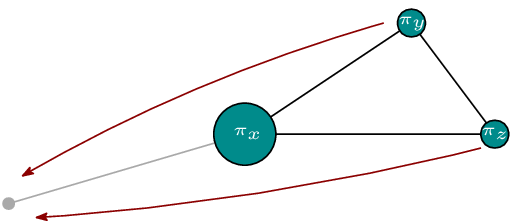}}
	\caption{\small Axiom 1.}\label{fig:Axiom1}
\end{figure}

Suppose in $\network$ there is a node with large group and there are two other smaller and equal-sized groups that are close to each other but further away from the larger group. Then network $\networkp$, in which smaller groups are joined at a node which is located in $\graph'$ at a distance equal to their average distance (in $\graph$) to the large group, is more polarized. Figure \ref{fig:Axiom1} illustrates such moves.\footnote{Note that $\graph$ and $\graph'$ do not have to be different and in our depictions we present axioms on the same graphs.}     Note that the distance of the fourth node to smaller nodes is not restricted in the axiom, allowing for moves such as the one depicted in panel (b) of Figure \ref{fig:Axiom1}.

 \begin{customax}{2}\label{ax:Axiom2}
	\textit{Data:} Network $\network$ with $n \geq 3$ nodes such that $\mass_x > \mass_z >0$, $\mass_y>0$, and $\mass_i = 0, \forall i \in \nodes(\graph)\setminus \{x,y,z\}$. Furthermore, $d_{\graph}(x,z) >d_{\graph}(x,y) > d_{\graph}(y,z)$.   \\
	\textit{Statement:} There exists $\epsilon >0$ such that for any network $\networkp$ with $(\mass'_{x'},  \mass'_{y'}, \mass'_{z'})=(\mass_{x},  \mass_{y}, \mass_{z})$, and $\mass'_{i'} = 0, \; i' \in \nodes(\graph') \setminus\{x', y',z'\}$  such that $d_\graph(x,z) = d_{\graph'}(x', z')$, $ 0< d_{\graph'}(x', y') -d_{\graph}(x,y) =d_{\graph}(y,z) - d_{\graph'}(y',z') < \epsilon $  we have $\mathcal{P}\networkp> \mathcal{P}\network$.
\end{customax}

 Axiom \ref{ax:Axiom2} applies when the group at one extreme is larger than the one at the other extreme and a third group is closer to the smaller of these two. When the group in-between moves slightly closer to the smaller group and away from the larger group, polarization increases.\footnote{Axiom \ref{ax:Axiom2} is rather weak as it applies to only those (small) moves such that an increase in distance from one extreme is equal to a decrease in the distance to the other extreme.} Note that the relative size of the group in the middle is not restricted.  Figure \ref{fig:Axiom2} illustrates such moves.

\begin{figure}[H]
	\centering
	\subcaptionbox{The move shown by the arrow increases polarization.\label{fig:ax2}}
	{\includegraphics[scale=1.5]{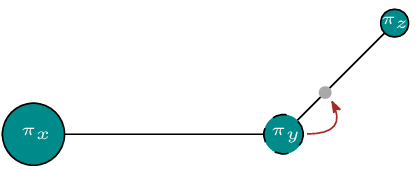}}
	\;\;\;\;\;        \;\;\;\;\;        \;\;\;\;\;
	\subcaptionbox{Axiom 2 applies in such a move as well, which is not possible on the real line.\label{fig:ax22}}
	{\includegraphics[scale=1.5]{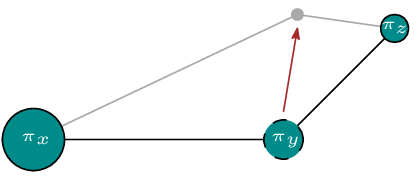}}
	\caption{\small Axiom 2.}\label{fig:Axiom2}
\end{figure}

 Note that the described move makes the middle group closer to the smaller group, but its new location does not have to be close to its original position, as seen in panel (b) of  Figure \ref{fig:Axiom2}. This kind of a move is not possible on the real line.

 \begin{customax}{3}\label{ax:Axiom3}
	\textit{Data:} Network $\network$ with $n \geq 3$ nodes such that $\mass_x >0,$ $\mass_y = \mass_z>0$ and $\mass_i = 0$ $\forall i \in \nodes(\graph)\setminus \{x,y,z\}$. Furthermore, $d_{\graph}(x,y) = d_{\graph}(x,z)=d>0$.   \\
	\textit{Statement:} For any $\Delta \in (0, \frac{\pi_x}{2}]$ and any network $\networkp$ with $(\mass'_{x'},  \mass'_{y'}, \mass'_{z'})=(\mass_{x}- 2 \Delta,  \mass_{y}+\Delta, \mass_{z}+\Delta)$, and $\mass'_{i'} = 0, \; i' \in \nodes(\graph') \setminus\{x', y',z'\}$  such that $d_{\graph'}(x',y') =d_{\graph'}(x',z')=d $  and $d_\graph(y,z) = d_{\graph'}(y', z')$,   we have $\mathcal{P}\networkp> \mathcal{P}\network$ whenever $d_\graph(y,z) =cd$, for any $c>1$.
	\end{customax}

Axiom \ref{ax:Axiom3} states that as long as the distance between two lateral groups is greater than the distance between the ``middle group" and a lateral group, a network in which  individuals from the group in the middle are reallocated to extreme points will exhibit higher polarization.  Note that the relative size of the group in node $x$ is not restricted. Furthermore,  in a network, $d_{\graph}(x,y) = d_{\graph}(x,z)=d$ implies only that $d_{\graph}(y,z) \leq 2d$, whereas on the real line $y\neq z$ and $|x-y|=|z-x|=d$ imply that $|z-y|=2d$. We will come back to this crucial point in Section \ref{sec:Discussion}.

\begin{figure}[th]
	\centering
	{\includegraphics[scale=2.2]{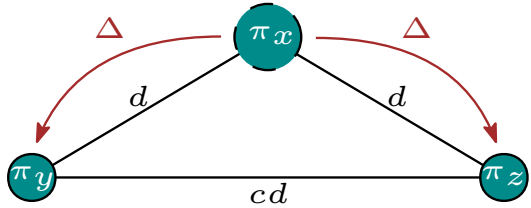}}
	\vspace{-2mm}
	\caption{\small Axiom 3 dictates that the dissolution of the middle group into two extreme nodes (with $c>1$) increases polarization.}
	\label{fig:ax3}
\end{figure}

We are now ready to state our central result, which identifies the measures of polarization in networks that satisfy Axioms 1--3. 

\begin{theorem}\label{thm:ThmMain}
A polarization measure $\mathcal{P}$ of the family defined in \eqref{eq:PolarizationGeneral} satisfies Axioms 1--3  and homotheticity if and only if 
\begin{align}\label{eq:Polarization}
\mathcal{P}\network = K \sum_{i \in \nodes(\graph)} \sum_{j \in \nodes(\graph)} \pi_i^2 \pi_j d_{\graph}(i,j), 
\end{align}
for some constant $K >0$. 
\end{theorem}

\begin{proof}\text{ }\newline
	\textit{Sufficiency.} Without loss of generality  set $K=1$.  We prove that  Axiom \ref{ax:Axiom1} and Axiom \ref{ax:Axiom2} are satisfied for
	\begin{align}\label{eq:PolarizationAlpha}
	\mathcal{P}_\alpha\network = K\sum_{i \in \nodes(\graph)}\sum_{j \in \nodes(\graph)}\pi_i^{1+\alpha} \pi_{j}d_{\graph}(i,j),
	\end{align}
	whenever $\alpha >0$.  Clearly,  \eqref{eq:PolarizationAlpha} becomes \eqref{eq:Polarization} when $\alpha=1$. Establishing this claim  for $\alpha \neq 1$ is important for the proof of Theorem \ref{thm:ThmGen}.\\
	
	\textit{Axiom \ref{ax:Axiom1}.}
	Let $\mass_x =p$ and $\mass_y = \mass_z=q$.  Using $\distance_{\graph'}(x',w') = \frac{\distance_{\graph}(x,y)+\distance_{\graph}(x,z)}{2}$ we get that
	\begin{align*}
	\Pol_{\alpha}\network = p^{1+\alpha}q \distance_{\graph}(x,y) + p^{1+\alpha}q \distance_{\graph}(x,z)  + 2 q^{1+\alpha}q \distance_{\graph}(y,z) + q^{1+\alpha}p \distance_{\graph}(x,y) + q^{1+\alpha}p \distance_{\graph}(x,z),
	\end{align*}
	while
	\begin{align*}
	\Pol_{\alpha}\networkp = 
	  p^{1+\alpha}(2q)\frac{\distance_{\graph}(x,y)+\distance_{\graph}(x,z)}{2} + (2q)^{1+\alpha} p \frac{\distance_{\graph}(x,y)+\distance_{\graph}(x,z)}{2}.
	\end{align*}
	After simplification we get:
	\begin{align*}
	&\Pol_{\alpha}\network = (\distance_{\graph}(x,y)+\distance_{\graph}(x,z))(p^{1+\alpha} q + q^{1+\alpha}p)+2 q^{2+\alpha}\distance_{\graph}(y,z), \text{ and }\\
	& \Pol_{\alpha}\networkp= (\distance_{\graph}(x,y)+\distance_{\graph}(x,z))(p^{1+\alpha} q + q^{1+\alpha}p) + (2^{\alpha}-1)(\distance_{\graph}(x,y)+\distance_{\graph}(x,z))q^{1+\alpha}p,
	\end{align*} 
	which implies 
	$\Pol_{\alpha}(\graph',\vecmass') > \Pol_{\alpha}\network$ whenever 
	%
	$(2^{\alpha}-1)(\distance_{\graph}(x,y)+\distance_{\graph}(x,z))p > 2 q \distance_{\graph}(y,z).$
	When $\distance(y,z)$ is small enough ($\distance(y,z)<\epsilon$) this inequality will hold for any $\alpha >0$ and $q$ small enough relative to $p$ ($q < \mu p$), {as required by Axiom \ref{ax:Axiom1}}.\\
	
	\textit{Axiom \ref{ax:Axiom2}.} Let $\mass_x =p$, $\mass_y = q$, and $\mass_z=r$.
	Subtracting we get:
	\begin{align*}
	\Pol_{\alpha}\network-\Pol_{\alpha}\networkp =&q^{1+\alpha}[p(d_{\graph'}(x',y')-d_{\graph}(x,y))+r(d_{\graph'}(y',z')-d_{\graph}(y,z))]+\\
	&q[p^{1+\alpha}(d_{\graph'}(x',y')-d_{\graph}(x,y))+r^{1+\alpha}(d_{\graph'}(y',z')-d_{\graph}(y,z))],
	\end{align*}
	which is positive for any $\alpha >0$ whenever $r<p$,  since $d_{\graph'}(x',y')-d_{\graph}(x,y) = d_{\graph}(y,z) - d_{\graph'}(y',z')$, {and therefore $\Pol_{\alpha}$ satisfies Axiom \ref{ax:Axiom2}.}\\
	
	\textit{Axiom 3.} 
	{We now show that $\Pol$ satisfies Axiom \ref{ax:Axiom3}. To this end let} $d_{\graph}(x,y) = d_{\graph}(x,z) = d $, and let $d_{\graph}(y,z) = cd$ with $c > 1$. Furthermore, let $\mass_x =p+2\Delta$ and $\mass_y =\mass_z=q-\Delta$.  
	We can write:
	\begin{align}\label{eq:FuncAxiom3}
	\Pol_\alpha(\network; \Delta) = 2cd \big(	(q-\Delta)^{2+\alpha}\big)+2d\left[(p+2\Delta)(q-\Delta)\big((p+2\Delta)^\alpha+(q-\Delta)^\alpha\big) \right].                 
	\end{align}
	To prove that $\Pol$ satisfies Axiom \ref{ax:Axiom3} it is sufficient to show that  $\left. \frac{\partial \Pol_\alpha(\network, \Delta) }{\partial \Delta}\right\vert_{\Delta=0, \alpha=1} < 0$ for every $(p,q) \gg 0$, except for at most one ratio $p/q$.
	Differentiating \eqref{eq:FuncAxiom3}  at $\Delta = 0$ and dividing by $2d$ $(\geq 0)$ we get:
	\begin{align*}
	\left. \frac{\partial \Pol_\alpha}{\partial\Delta}\right\vert_{\Delta=0}<0 \iff -p^\alpha\big(p-2(1+\alpha)q\big)+q^\alpha\big(-(1+\alpha)p+2q-(2+\alpha)cq\big)<0.
	\end{align*}
	Dividing by $p^{1+\alpha}>0$ and using notation $z = q/p$ we get:
	\begin{align*}
	\left. \frac{\partial \Pol_\alpha}{\partial\Delta}\right\vert_{\Delta=0}<0 %
	\iff f(z, \alpha, c) <0,
	\end{align*}
	where
	$f:\Realo^2\times[1,2] \rightarrow \Real $ is defined with:
	\begin{align}\label{eq:Funf}
	f(z, \alpha, c) = (1+\alpha) \left [ z- \frac{z^{\alpha}}{2} + \frac{z^{1+\alpha}}{2} \frac{(2-c(2+\alpha))}{1+\alpha} \right] - \frac{1}{2}.
	\end{align}

	 Proving that $\left. \frac{\partial \Pol_\alpha}{\partial\Delta}\right\vert_{\Delta=0, \alpha=1}<0$ for any $c \in (1, 2]$ (except for at most one ratio $p/q$) is equivalent to proving that $f(z,1,c)<0$ for any $c \in (1, 2]$ (except for at most one point $z$). One can easily verify that $f(z,1,c)<0$ ($f(z,1,c)$ is a quadratic function in $z$) for any $c \in (1, 2]$, therefore $\Pol$ satisfies Axiom \ref{ax:Axiom3} as well.\footnote{When $c=1$ then $f(z,1,1) \leq 0$ where the equality holds only at point $z=1$.}

	\textit{Necessity.} The proof is analogous to the proof of Theorem 1 in ER.  We describe it briefly, and refer the reader to ER for detailed derivation. Axioms 1--2 imply that function  $T$ is linear in its second argument, thus $\theta(\mass, \delta) \equiv  T(I(\mass), a(d_{\graph}(i,j)))$ can be written as  $\theta(\mass, \delta)= \phi(\pi)\delta$. Furthermore, Axiom \ref{ax:Axiom1} implies that $\phi(\cdot)$ is an increasing function.\footnote{See \cite{kawada2018characterization} for a solution to a technical problem arising from the original formulation of Axiom 1 in ER.}  Homotheticity implies that $\phi(\pi)= K \pi^{\alpha}$ for some constants $(K,\pi) \gg 0$. 
	
	Finally, Axiom \ref{ax:Axiom3} implies that $f(z,\alpha,c) \leq 0$,  with equality holding at most at one point $z$. In the first part of the proof we established that when $c>1$ and $\alpha=1$, $f(z,\alpha,c) <0 $ for all $z>0$. Lemma \ref{lem:AlphaNeq1}  implies that for any $\alpha \neq 1$  we can find $c>0$ such that $f(z,\alpha,c)>0$, which concludes the proof.
\end{proof}

A few comments are in order. First, recall that ER characterize measures of polarization on the real line as
\begin{align}\label{eq:ERPolarization}
P^{ER}(\vec{\pi}) = K \sum_{i=1}^{n}\sum_{j=1}^n \pi_i^{1 + \alpha} \pi_j |i-j|, 
\end{align}
with $K>0$ and  $\alpha \in (0, \alpha^*]$, with $\alpha^* \simeq 1.6$. The main difference between \eqref{eq:Polarization}  and \eqref{eq:ERPolarization}  is that  the index in \eqref{eq:Polarization} implies $\alpha=1$. The reason for this difference lies in the nature of the distances, discussed in relation with Axiom \ref{ax:Axiom3}. It requires that a move from a middle mass ($\mass_x$)  to the lateral points ($\mass_y$ and $\mass_z$) equidistant from the middle  increases  polarization whenever they are individually further away from each other than they are to the midpoint. Contrary to the real line, in $\network$, $d_{\graph}(y,z)$ is not determined by $d_{\graph}(x,y) = d_{\graph}(x,z)$, and in fact it can  very well happen that $d_{\graph}(y,z) < d_{\graph}(x,y)$ even when $d_{\graph}(x,y) = d_{\graph}(x,z)$. We revisit this important matter in Section \ref{sec:Ax3} below. Note that Axioms \ref{ax:Axiom1} and \ref{ax:Axiom2} also require adaptation for the network setup, but these adaptations are minor and  do not have important implications on the form of the characterized family of measures.

Intuitively, a society is polarized if it can be grouped in a small number of homogeneous groups of similar sizes that are very different from each-other and polarization is often conceptualized to capture the level of bipolarity (or bimodality).\footnote{See \cite{foster2010polarization} for a discussion on bipolarity of income distributions and \cite{dimaggio1996have} for a more general discussion on bimodality, among others.} Thus, it is desirable that a polarization measure is maximized at a bipolar distribution. A bipolar network is one where the population is split equally into two extreme (most distant) nodes. The maximal distance between two nodes  in graph $\graph$ is called the \textit{diameter} of $\graph$ and is denoted by $d(\graph)$.\footnote{More formally, $d(\graph) = \max_{i,j  \in N(\graph)} d_{\graph}(i,j)$. See \cite{vega2007complex} or \cite{jackson2008}. } For any graph $\graph$ let $\vecmass^{B}(\graph)$ denote the distribution in which the population is split equally across two nodes at distance $d(\graph)$.
Our next result shows that  $(\graph, \mass^{B}(\graph))$ is more polarized than any other network $\network$ under any measure within our characterization.

 \begin{proposition} \label{prop:MaxElement}
 $\mathcal{P}(\graph, \vecmass^B(\graph))> \mathcal{P}(\graph, \vecmass)$ for any $(\graph,\vecmass)$ with  $\vecmass\neq \vecmass^B(\graph)$ and any measure $\mathcal P$ defined in \eqref{eq:Polarization}. 
 \end{proposition}
 
  \begin{proof}\textit{ }\newline	
 	We first prove that	for any network $\network$ such that $\vecmass$ has at lest $4$ nonzero mass points,  there exists a $3$ node network $(\graph^*,\vecmass^*)$ {with} $g^{*}_{ij} = d(\graph)$ for $i, j \in \nodes(\graph^*)$, and  $\sum_{i=1}^3 \pi_i^* =\sum_{i \in \nodes{\graph}} \pi_i$ such that $\Pol\network < \Pol(\graph^*, \vecmass^*).$ 
 	
 	The proof is constructive. Assume, without loss of generality, that in $\network$, we have $\pi_1 \geq \pi_2 \geq \dots \geq \pi_n$ with $\pi_k>0$ and $\pi_{k+1}=0$ for some $k\geq 4$. Fixing $K=1$ in \eqref{eq:Polarization} (without loss of generality) we get:
 	\begin{align}\label{ineq:PolMaxStep1}
 	\begin{split}
 	\Pol \network \leq& d(\graph) \sum_{i=1}^{k} \sum_{j=1}^{k}\pi_i^2 \pi_j d_{\graph}(i,j)
 	\\
 	=& d(\graph) \left[\sum_{i=1}^{k-2} \sum_{\substack{j=1\\ j \neq i}}^{k-2}\pi_i^{2}\pi_j + \pi_{k-1}^2\sum_{\substack{j=1\\j \neq k-1 }}^{k} \pi_j + \pi_{k}^2\sum_{\substack{j=1\\j \neq k }}^{k} \pi_j+\pi_{k-1}\sum_{j=1}^{k-2} \pi_j^2 + \pi_{k}\sum_{j=1}^{k-2} \pi_j^2 \right].
 	\end{split}
 	\end{align}		
 	
 	Denote the right hand side expression in \eqref{ineq:PolMaxStep1} with $\Pol(\graph',\vecmass')$, where $g'_{ij}=d(\graph)$ for all $i,j\in N(\graph)$ and $\pi'_i=\pi_i$ for all $i\in N(\graph)$. Consider now a change in $(\graph',\vecmass')$ such that masses in nodes $k$ and $k-1$ are merged at one of these nodes to obtain $(\graph'',\vecmass'')$. Simple algebra gives: 
 	\begin{align*}
 	\Pol(\graph'',\vecmass'') = d(\graph)\left[\sum_{i=1}^{k-2} \sum_{\substack{j=1\\ j \neq i}}^{k-2}\pi_i^{2}\pi_j + (\pi_{k-1} + \pi_{k})^2 \sum_{j=1}^{k-2} \pi_{j} + (\pi_{k-1} + \pi_{k}) \sum_{j=1}^{k-2} \pi_{j}^2 \right].
 	\end{align*}	
 	Subtracting $\Pol(\graph',\vecmass')$ we get:
 	\begin{align}\label{ineq:PolMaxStep2}
 	\begin{split}
 	\Pol(\graph'',\vecmass'')-\Pol(\graph',\vecmass')
 	=d(\graph)\pi_{k-1}\pi_{k} \left[ 2 \sum_{j=1}^{k-2}\pi_j - (\pi_{k-1} + \pi_{k})\right] >0,
 	\end{split}
 	\end{align}	
 	where the inequality follows from the choice of $k$ and $k-1$ and the fact that $k \geq 4$. 
 	Thus, for any network $\network$  with $|N(\graph)|\geq 4$ have $\Pol\network < \Pol(\graph'',\vecmass'').$
 	
 	If $k=4$, $(\graph^*, \vecmass^*) = (\graph'',\vecmass'')$. If $k>4$, the above described procedure of joining the masses in nodes $k-1$ and $k-2$ can be iteratively applied.

 	To conclude the proof of the proposition, consider 3 different cases for $\network$:
 	\begin{itemize}
 		\item [(i)]  $|\{i\in\nodes(\graph):\pi_i>0\}|=2$. Clearly $\Pol(\graph,\vecmass^B(\graph)) =
 		d(\graph)2\left(\frac{\pi_i+\pi_j}{2}\right)^3>d_{\graph}(i,j)\big(\pi_i\pi_j(\pi_{i}+\pi_j)\big)=
 		\Pol(\graph,\vecmass)$ for any $\network \neq (\graph, \vecmass^{B})$.
 	
 		\item[(ii)]  $|\{i\in\nodes(\graph):\pi_i>0\}|=3.$  We consider two cases.
 		\begin{itemize}
 			\item[(a)] If $\pi_i = \pi_j =\pi_k >0$,  and $\pi_{\ell}=0$ for all $\ell \in \nodes(\graph)\setminus \{i,j,k\}$. One can directly check that in this case $ \Pol(\graph, \vecmass^{B}(\graph))> \Pol\network $. 
 			\item[(b)] $ \pi_i,\pi_j ,\pi_k >0 $ and $\pi_{\ell}=0$ for all $\ell \in \nodes(\graph)\setminus \{i,j,k\}$ and not all nonzero masses are equal. Suppose, without loss of generality, that  $\pi_{i} \geq \max\{\pi_j, \pi_k\}.$ Consider network $(\graph', \vecmass')$ such that $N(\graph') = N(\graph)$, $g'_{ij} = d(\graph) \; \forall i,j \in N(\graph')$,  with  $\pi'_i = \pi_i$, $\pi_j'= \pi_j + \pi_k$, and $\pi'_{\ell}=0, \; \ell \in \nodes(\graph')\setminus \{i,j\}$.
 			It can be directly checked that $2\pi_i > \pi_j + \pi_k$ implies  $\Pol(\graph', \vecmass') > \Pol(\graph, \vecmass)$.  The claim follows from the fact that $\Pol(\graph', \vecmass^{B}(\graph')) \geq \Pol(\graph', \vecmass')$.
 		\end{itemize}
 
 		
 		\item [(iii)]  $|\{i\in\nodes(\graph):\pi_i>0\}|\geq4.$ The claim follows  from the first part of the proof and (ii).
 	\end{itemize} 	
 	\vspace{-25pt}	
 \end{proof}
%
\vspace*{-5pt}
\section{Discussion} \label{sec:Discussion}

In this section, we first discuss some important properties of the measures we characterize in relation to the structure of networks. Then we show  how our work is related to previous papers in the literature. We conclude this section with a discussion on how the weakening of the Axiom 3 can relate our characterization to the one in ER, by exactly describing the relationship between the importance of identification ($\alpha$) and the network structure. 

\vspace*{-5pt}

\subsection{Network structure and polarization}
We first want to emphasize that the structure of a graph $\graph$ determines the distance between any two nodes in $N(\graph)$. A change in the structure of a graph $\graph$, \emph{e.g.,} deleting a link, may affect the measured levels of polarization, even if $\vecmass$ stays the same. Although empty (zero-weight) nodes do not directly contribute to the level of polarization, they may  be important ``indirectly" if, for instance, they are located on the shortest path between some non-empty nodes. Figure \ref{fig:Disc} illustrates this point.

\begin{figure}[H]
	\centering
	\subcaptionbox{$(\graph, \vecmass)$.\label{fig:Desc1}}
	{\includegraphics[scale=1.4]{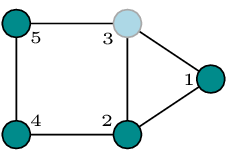}}
	\;\;\;\;\;        \;\;\;\;\;        
	\subcaptionbox{$(\graph', \vecmass)$.\label{fig:Desc2}}
	{\includegraphics[scale=1.4]{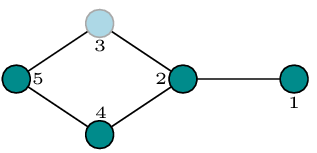}}
	\;\;\;\;\;        \;\;\;\;\;       
	\subcaptionbox{$(\graph'', \vecmass'')$.\label{fig:Desc3}}          
	{\includegraphics[scale=1.4]{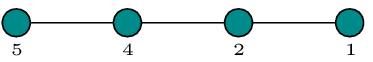}}
	\caption{\small \small Three networks where each link has weight $1$ and each node except the node 3 has weight $1$ ($\pi_3=0$). $(\graph', \vecmass)$ is obtained from $(\graph, \vecmass)$ by deleting the link $g_{13}$. $(\graph'', \vecmass'')$ is obtained from $(\graph, \vecmass)$ by deleting the node 3. Note that $\pi_i=\pi''_i$ for all $i\in N(\graph'')$. We have $\Pol\network < \Pol (\graph', \vecmass) = \Pol(\graph'', \vecmass'')$.}
	\label{fig:Disc}
\end{figure}

Next, we want to note that given Proposition \ref{prop:MaxElement}, we have that $d(\graph)>d(\graph')$  implies $\mathcal{P}(\graph, \vecmass^{B}(\graph))> \mathcal{P}(\graph', \vecmass^{B}(\graph'))$. That is, comparing two bipolar networks, the larger the diameter, the higher the polarization. 

Finally, in the special case when $\vec{\pi} =\vec{1}$, $\mathcal{P}\network$ is proportional to the \textit{average shortest path} in the graph $\graph$.\footnote{The average shortest path in a network is closely related to the ``{closeness}" measure  \citep{vega2007complex,jackson2008}.} Thus, the closer the individuals are, on average, the less polarized the network is. 

\vspace*{-5pt}
\subsection{Relation to previous results}
We argue that the settings considered in ER and MRQ are special cases of our setting, and hence our results can be seen as generalizations of theirs. To start with, recall that ER consider distributions on the real line with a {finite} support (p. 830). It is straightforward to note that any distribution as such can be described as a network. To see this, let $\vecmass$ be a distribution with a set of $N$ mass points. Consider graph $\graph$ with $N$ nodes such that $g_{ij} = |i-j|$ for any two adjacent mass points $i$ and $j$ on the real line, and $ij \notin \graph$ otherwise.
\footnote{This is not the unique way to represent a discrete distribution with $n$ mass points as a network. 
However, any consistent  representation that relies on the same metric will lead to a network with the same polarization.} Indeed, we can represent any distribution on an $m-$dimensional space with finite number of mass point as a network by simply setting  $g_{\vec{ij}} =\norm{\vec{i}-\vec{j}}$, where $\norm{\cdot}$ can be any norm.

In the setting considered in MRQ the distance between \textit{any two} different groups equals to 1. 
It is immediate to note that this setting can be described {by the network $(\graph,\vecmass)$  where $\graph$ is \textit{the complete graph} ($g_{ij} = 1$ for any pair of different nodes $i, j \in N(\graph)$)}. MRQ proposes a different set of axioms.\footnote{The logical dependence between our axioms and the ones in MRQ is an interesting question that is left for future research.} Our setting is more general than the one in MRQ in that it allows considering graphs that are not \textit{complete}, with links that have \textit{different weights}. Moreover, some studies, including \cite{desmet2009linguistic} and \cite{dower2017colonial},  empirically  contrast the ER measure (with $\alpha=1$)  with the MRQ measure in situations where distances between groups are observed and non additive (\emph{i.e.,} ethnolinguistic distance). This is a setting that can be described using our model but is not within the original ER setup.  In \cite{desmet2009linguistic} and \cite{dower2017colonial}  the ER measure is both economically and statistically  significant (at conventional levels) when examining the effect of polarization on redistribution and conflict respectively. At the same time, the MRQ measure, which imposes that each group is at the same distance from any other group, is not significant in explaining the same outcomes. Therefore, accounting for distances between groups, or at least their proxies, is empirically important as well. Our paper is the first to provide formal justification to use the ER polarization measure with $\alpha=1$ for the measurement of polarization in such contexts. 
 
\vspace*{-5pt}
\subsection{Axiom 3 and its role in the network setting}\label{sec:Ax3}

 Axiom \ref{ax:Axiom3} requires that the described change in $\network$ leads to an increase in polarization only when the distance between lateral nodes is at least as large as the distance between the center node and lateral nodes. We now discuss less demanding versions of Axiom \ref{ax:Axiom3}, labeled systematically as Axiom \ref{ax:Axiom3Prime}, in which we require that the scenario  in  Axiom \ref{ax:Axiom3} leads to an increase in polarization only if the lateral nodes are ``{far enough}" (quantified by the scalar $c$) from each other. This is of interest also because some settings imply a specific network structure in which there is a clear lower bound for the distance between two lateral nodes contemplated in Axiom \ref{ax:Axiom3}. For instance, as we saw before, any discrete distribution on the real line can be represented with a line network. On any line network, the distance between lateral nodes is the double of the distance between the middle node and a lateral node, as it is on the real line.
 
 \begin{customax}{3$(c)$}\label{ax:Axiom3Prime}
	\textit{Data:} Network $\network$ with $n \geq 3$ nodes, $\mass_x > \mass_y = \mass_z>0$ and $\mass_i = 0$ for all $i \in \nodes(\graph)\setminus \{x,y,z\}$. Furthermore, $d_{\graph}(x,y) = d_{\graph}(x,z)=d>0$.   \\
	\textit{Statement:} Fix $c \in (1,2]$. For any $\epsilon \in (0, \pi_x]$ and any network $\networkp$ with $(\mass'_{x'},  \mass'_{y'}, \mass'_{z'})=(\mass_{x}-\epsilon,  \mass_{y}+\frac{\epsilon}{2}, \mass_{z}+\frac{\epsilon}{2})$, and $\mass'_{i'} = 0, \; i' \in \nodes(\graph') \setminus\{x', y',z'\}$  such that $d_{\graph'}(x',y') =d_{\graph'}(x',z')=d $  and $d_\graph(y,z) = d_{\graph'}(y', z')$,   we have $\mathcal{P}\networkp> \mathcal{P}\network$ whenever $d_\graph(y,z) \geq  d c$.
\end{customax}

\begin{figure}[H]
       \centering
         {\includegraphics[scale=2.2]{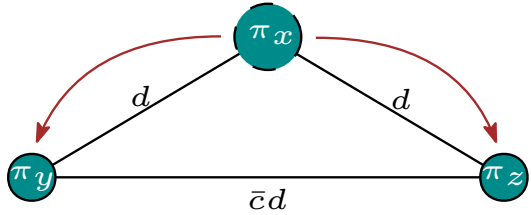}}
           \vspace{-2mm}
       \caption{\small Axiom \ref{ax:Axiom3Prime} {requires} that the move shown by the arrows should increase polarization if $\bar c \geq c$, for a fixed $c \in (1,2]$ .}
       \label{fig:ax33}
     \end{figure}  
     
When $c =1$, we have the  same statement  as in Axiom \ref{ax:Axiom3}, while for $c =2$ we have essentially the Axiom 3 in ER.  The particular  value of $c$ has important implications on the resulting measure of polarization, as stated in Theorem \ref{thm:ThmGen}.

\begin{theorem}\label{thm:ThmGen}
 Fix $c \in (1,2]$. There exists an interval $[\underline\alpha (c), \bar{\alpha}(c)]\subseteq(0,\alpha^*]$ with $\alpha^*\simeq 1.6$ such that the polarization measure $\mathcal{P}$ of the family defined in \eqref{eq:PolarizationGeneral} satisfies Axioms 1, 2, and \ref{ax:Axiom3Prime}  and homotheticity if and only if 
	\begin{align}
	\mathcal{P}_\alpha\network = K\sum_{i \in \nodes(\graph)}\sum_{j \in \nodes(\graph)}\pi_i^{1+\alpha} \pi_{j}d_{\graph}(i,j)
	\end{align}
	 for some constant $K>0$ whenever $\alpha\in[\underline\alpha(c), \bar{\alpha}(c)]$.  Furthermore,    $c_2 >c_1 \implies [\underline\alpha (c_1), \bar{\alpha}(c_1)] \subset [\underline\alpha (c_2), \bar{\alpha}(c_2)] $. 
\end{theorem}

\begin{proof}\textit{ }\\
		See the proof of Theorem \ref{thm:ThmMain} for the proofs of claims regarding Axiom \ref{ax:Axiom1} and Axiom \ref{ax:Axiom2} (the \textit{Sufficency} and  the \textit{Necessity} part). Similarly, Axiom \ref{ax:Axiom3Prime} holds iff  $\alpha$ is such that $f(z, \alpha, c) < 0$ except for at most one point $z$, where $f$ is defined in \eqref{eq:Funf}. To conclude the proof, two observations about  $v(\alpha, c) = \max_{z \geq 0}f(z, \alpha, c)$ are important. First, $v$ is increasing in $\alpha \in (1,2]$ for any fixed $c \in (1,2]$ and changes the sign on the considered interval. Thus, there exists $\bar{\alpha}(c)$ such that $v(\alpha, c)  \leq 0$ for $\alpha \in (1, \bar{\alpha}(c)]$.  Since $v(\alpha, c)$ is decreasing in $c$,  $\bar{\alpha}(c)$ is increasing in $c$. Second, for $\alpha <1$ and fixed $c$,  $v$ decreases in $\alpha$ whenever $v(\alpha, c) \geq 0$ eventually becoming negative as $v(1,c) <0$ for $c>1$. This implies the existence of  $\underline{\alpha}(c) \in [0,1]$. Since $v$ decreases in $c$ we have that  $\underline{\alpha}(c)$ increases in $c$. From these two observations\footnote{{See Lemma \ref{lem:AlphaGeq1} and \ref{lem:AlphaLeq1} in Appendix \ref{Ap:Proofs} for the formal statements and proofs of these two observations.}} we conclude $2 \geq c_2 >c_1>1 \implies [\underline\alpha (c_1), \bar{\alpha}(c_1)] \subset [\underline\alpha (c_2), \bar{\alpha}(c_2)] $. 
\end{proof}

Theorem \ref{thm:ThmGen} shows that as we make Axiom \ref{ax:Axiom3} less demanding,  the range of values of parameter $\alpha$ for which our axioms is satisfied expands monotonically. In particular, if we restrict ourselves to line networks, then the network structure implies that any move described in Axiom \ref{ax:Axiom3} is consistent with Axiom \ref{ax:Axiom3Prime}  for $c=2$, and  Axioms \ref{ax:Axiom1}, \ref{ax:Axiom2} and \ref{ax:Axiom3Prime} can be seen as restatements of the Axioms 1--3 in ER. 

Finally, it should be noted that the claim in Proposition \ref{prop:MaxElement} holds {only} for measures characterized in Theorem \ref{thm:ThmMain}, {and not for any other measure as in \eqref{eq:PolarizationAlpha} with $\alpha \neq 1$.} To see this, take any graph $\graph$ such that $N(\graph) =\{x,y,z\}$ with $0< g_{xy} = g_{xz} \leq g_{yz}$. Then 
for any $\alpha \in \Realo \setminus \{1\}$, there exists a distribution $\vecmass \neq \vecmass^{B}(\graph)$ and $\epsilon >0$  such that $\mathcal{P} (\graph, \vecmass)> \mathcal{P}_{\alpha}(\graph, \vecmass^{B}(\graph))$ whenever $g_{yz} = g_{xz} + \epsilon$. This is a direct consequence of the fact that for $\alpha\neq1$, $\mathcal{P}_{\alpha}$ does not satisfy Axiom \ref{ax:Axiom3} when $c$ is arbitrary close to 1.

\vspace*{-5pt}
\section{Conclusion}\label{sec:Conclusion}

{We have introduced a model of polarization in networks. This model can be used to study the levels and trends of polarization in a wide range of applications. In Section \ref{sec:PolMeasure}, we discussed several examples from political processes in parliaments and public preferences. The potential of our proposal is by no means restricted to these examples as pointed to before. To name a few areas beyond the domain of polity, for which a recent survey is provided by \cite{battaglini2019social}, \cite{bail2016combining} constructs weighted networks between advocacy organizations based on the frequency of words in the shared vocabulary of their posts. \cite{stewart2018examining} construct retweet networks to study the impact of suspicious troll activity on the levels of polarization on Twitter \citep[see, also,][]{conover2011political}. \cite{farrell2016corporate} constructs a network of organizations based on the activities of affiliates to study polarization on climate change issues among organizations. \cite{o2018scientific} propose the network formalism to study polarization in scientific communities around beliefs based on scientific knowledge.  \cite{difonzo2013rumor} employ a network-based approach on capturing polarization of rumor beliefs in the context of social impact theory.}

Reconstructing the axiomatic analysis of ER, we characterized a family of measures within our model. Importing the axiomatic approach needs a careful attention due to the distinct nature of the geodesic distance on networks compared to the Euclidean distance on the real line. Our characterization result shows that the class of measures characterized by ER carries almost intact to the networks. The only bite is in the value of the parameter for the effect of identification on effective antagonism. We find that $\alpha=1$ is a necessary and sufficient condition for the measures of polarization in the form of aggregate antagonisms to satisfy the aforementioned axioms, together with hometheticity. We demonstrate that polarization is maximized when the population is allocated on the two most distant nodes in the network. Finally, we discuss how restricting to specific class of network structures may expand the class of polarization measures.

Our model can be further developed along different dimensions. One promising avenue for future research pertains to extending the measures so as to capture the intra-group heterogeneity, which could also be described as a network. In that case, the identification function should additionally depend on the within-group structure. Another direction for future \label{key}research  concerns the existence of interesting characterizations outside the identification-alienation framework but with the same axioms, as these two are independent.

\newpage
\bibliographystyle{apalike}
\bibliography{Literature_Polarization.bib}

\newpage
\appendix

%

\section{Appendix: Proofs}\label{Ap:Proofs}
\onehalfspacing
In what follows, we denote the maximal value of parameter $\alpha$ in ER  with $\alpha^*$ (so that $\alpha^*\simeq 1.6$). 

In Lemmas \ref{lem:AlphaNeq1}, \ref{lem:AlphaGeq1}, and \ref{lem:AlphaLeq1} we show some properties of function $f$ defined in equation \eqref{eq:Funf} that are invoked in the proofs of Theorem \ref{thm:ThmMain} and \ref{thm:ThmGen}.

  \begin{lemma}\label{lem:AlphaNeq1}
	For any $\alpha \neq 1$ there exists $c>1$ such that $f( z, \alpha, c)>0$.
\end{lemma}
\begin{proof}
	We consider two cases, $\alpha>1$ and $\alpha <1$. We show that in each of these cases we can find  $z$ (infinitely many of them) such that $f(z, \alpha,1) >0$. The continuity of $f$  in $c$ then implies  that this will also  be the case for $c>1$ that is close enough to 1.  
	\begin{itemize}
		\item [(i)] When $\alpha >1$ we focus on  $z \in \left(0,1 \right)$. For such $z$ and $\alpha$, we have $z>z^{\alpha}$, and the following holds:
		\begin{align*}
		f(z, \alpha, 1) =& (1+\alpha)\left[z - \frac{z^{\alpha}}{2} - \frac{z^{1+\alpha}}{2(1+\alpha)}\right] -\frac{1}{2} >  (1+\alpha)\left[z - \frac{z}{2} - \frac{z^{2}}{2(1+\alpha)}\right] -\frac{1}{2}\\
		=&-\frac{1}{2}(z-1)(\alpha z-1)
		\end{align*} 
		The last expression is positive for $z \in \left(\frac{1}{\alpha}, 1\right)$
		\item [(ii)] When $\alpha <1$ we focus on  $z \in \left(1, \infty \right)$. For such $z$ and $\alpha$, we have $z>z^{\alpha}$, and analogously to the previous case we conclude that 	$f(z, \alpha, 1) >0$ for $z \in (1, \frac{1}{\alpha})$.
	\end{itemize}
Points (i) and (ii) together with the continuity of $f$ in $c$ imply that for any $\alpha \neq 1$ we can find $c>1$ sufficiently close to 1 such that $f( z, \alpha, c)>0$

\end{proof}

\begin{lemma}\label{lem:AlphaGeq1}

Let $1 \leq \alpha \leq \alpha^* $  and $c \in (1,2]$.  There exists $\upalpha = \upalpha(c) \in (1,\alpha^*]$ such that $\max_{z \geq 0} f(z,\alpha,c) \leq 0$  whenever $\alpha \leq \upalpha $. Furthermore, $\upalpha$ is increasing in $c$.
\end{lemma}	

\begin{proof}[Proof of Lemma \ref{lem:AlphaGeq1}] 
We first note that when $\alpha \geq 1$ the value function $v(\alpha, c) = \max_{z \geq 0}f(z, \alpha, c)$ is strictly decreasing in $c$.
For $\alpha \geq 1$, $f$ is concave in $z$. Thus, the maximum of $f$ is given by the first order condition:
	\begin{align}\label{eq:FOCfWithc}
	\frac{1}{2} (1 + \alpha) \left(2 - \alpha z^{\alpha -1} + (2 - (2 + \alpha) c) z^{\alpha} \right)=0.
	\end{align} 
	Taking derivative of the value function $v(\alpha, c) = \max_{z \geq 0}f(z, \alpha, c)$ with respect to $c$, and applying the envelope theorem, we get:
	\begin{align*}
	\frac{\partial v}{\partial c} = \frac{\partial f}{\partial z}\frac{\partial z}{\partial c} +  \frac{\partial f}{\partial c} = \frac{\partial f}{\partial c} = -\frac{1}{2} (\alpha+2) z^{\alpha+1} <0,
	\end{align*}
	so the value function is (strictly) decreasing in $c$. This implies that $v(\alpha, c) >  v(\alpha,2)$, for any $c \in (1,2)$.
	
We show now that for any fixed $c \in (1,2]$,  $v(\alpha,c)$  changes sign from negative to positive when $\alpha$ increases from 1, and that $v(\alpha,c)$ is strictly increasing in $\alpha$  for $\alpha \geq 1$. 	
	
	We know from the observations on the function $f$ defined in \eqref{eq:Funf}  and related discussion in the proof of Theorem \ref{thm:ThmMain} that, for any $c>1$,  $v(1, c) <0$.   Since  $v(2,2) >0$, as pointed out in ER (p. 833) and $v(\alpha, c)$ is decreasing in $c$, it must be that $v(2,c) >0$ for any $c \in (1,2]$.  Therefore  $v(1,c) <0$ and $v(2,c) >0$ for any $c \in (1,2]$. To show that there exist $\upalpha(c)$ from the claim of the Lemma, we show that $v(\alpha, c)$ is increasing in $\alpha$ for values $\alpha \geq 1$. 
	Indeed:
	\begin{align*}
	\frac{\partial v}{\partial \alpha} = \frac{\partial f}{\partial z}\frac{\partial z}{\partial \alpha} +  \frac{\partial f}{\partial \alpha} = \frac{\partial f}{\partial \alpha} = \frac{1}{2} (2 z - z^{\alpha} (1 + c z + (1 + \alpha + (-2 + (2 + \alpha) c) z) \ln z)) >0.
	\end{align*} 

    To see that the above derivative is positive, first note that  the first order condition \eqref{eq:FOCfWithc} implies that at the maximum of $f$:
	\begin{align}\label{eq:zleq1_Lem2}
	z^{\alpha-1} =\frac{2}{\alpha -z\big(2-(2+\alpha)c\big)}. 
	\end{align}
	Equation \eqref{eq:zleq1_Lem2}  together with the fact that $\alpha -1 \geq 0$ and $c > 1$ implies that $z<1$. Indeed, if $z\geq 1$ then the RHS of \eqref{eq:zleq1_Lem2} would be greater than 1, while the LHS of   \eqref{eq:zleq1_Lem2} would be smaller or equal to 1, since the denominator $\alpha - z\big(2-(2+\alpha)c\big)$ would be  greater than 2 since  $(2+\alpha)c >3$ and $z\geq 1$.  Plugging \eqref{eq:zleq1_Lem2} into the expression for $\frac{\partial v}{\partial \alpha}$ from above. we get:
	\begin{align*}
		\frac{\partial v}{\partial \alpha} &=  \frac{1}{2} (2 z - z^{\alpha} (1 + c z + (1 + \alpha + (-2 + (2 + \alpha) c) z) \ln z)) \\
		&=\frac{1}{2} (2 z - \frac{2}{\alpha -z\big(2-(2+\alpha)c\big)} (1 + c z + (1 + \alpha + (-2 + (2 + \alpha) c) z) \ln z)) \\
		&=-z\frac{(1-\alpha) + (2- 3 c - 2 \alpha c)z + \left[1+ 2 \alpha + (-2 +  2c + \alpha c)z\right] \log z}{\alpha -z\big(2-(2+\alpha)c\big)},
	\end{align*}	
 which is clearly positive for $z<1$, $\alpha \geq 1$ and $c>1$.

	Therefore, the intermediate value theorem implies that, for any $c \in (1,2]$ there exist $\upalpha(c) \in (1,2]$ ( $v(1,c) <0$ for any $c>1$) such that $ \max_{z \geq 0} f(z,\alpha,c) \leq 0$ whenever $\alpha \leq  \upalpha(c)$ (with equality only when $\alpha = \upalpha$).
	
	Finally, $\frac{\partial v}{\partial c} <0$ and $\frac{\partial v}{\partial \alpha} >0$  imply that $\upalpha(c)$ increases with $c$ for $c \in (1,2]$, and hence $\upalpha(c) \leq \upalpha (2) = \alpha^*$.
\end{proof}

\begin{lemma}\label{lem:AlphaLeq1}
	Let $0 \leq \alpha \leq 1$ and $c \in (1,2]$. Either $\max_{z \geq 0} f(\alpha,z,c) < 0$ for all $\alpha \in [0,1]$ or there exists $\loalpha = \loalpha(c) \in [0,1]$ such that $\max_{z \geq 0} f(\alpha,z,c) \geq 0$ whenever $\alpha \leq \loalpha(c)$. Furthermore, $\loalpha$ is decreasing  in $c$.
\end{lemma}

\begin{proof}[Proof of Lemma \ref{lem:AlphaLeq1}]
	
	Let $\alpha \leq 1$. We first prove that $f(z,\alpha, c) \geq 0$ only if $z\geq 1$. Then we show that for $z \geq 1$  $f$ is strictly decreasing  in $\alpha$. Therefore
	in that case $v(\alpha, c) = \max_{z \geq 0} f(z, \alpha, c)$ is strictly decreasing in $\alpha$ as well.
	
To show that  $f(z,\alpha, c) \geq 0 \Rightarrow z \geq 1$ we show $z <1 \Rightarrow f(z,\alpha, c) <0$. Since $f(z,\alpha, 1) <0 \Rightarrow f(z,\alpha, c) <0$ ($f$ is decreasing  in $c$), it is sufficient to show that $z <1 \Rightarrow f(z,\alpha, 1) <0$.

	We have 
	\begin{align*}
	f(z,\alpha, 1) =& \frac{1}{2} \left(-1 + 2(1+\alpha)z -(1+\alpha)z^{\alpha} - \alpha z^{1+\alpha}\right)\\
	<&-1 +(1+\alpha)(2z-z)-\alpha z^{1+\alpha}\\
	<&-1 + (1+\alpha)z - \alpha z^2 =(1-\alpha z)(z-1)<0,
	\end{align*}
	where the inequalities follow the fact that $\alpha \leq 1$ and $z<1$.
	
	Next we prove that $f$ is decreasing in $\alpha$  when $z\geq 1$.  Differentiating we get:
	\begin{align*}
	\frac{\partial f}{\partial \alpha} = \frac{1}{2} \left[2z - z^{\alpha} - c z^{\alpha +1} - z^{\alpha} \left(1+ \alpha + (c(\alpha+2)-2\right) z )\ln z \right]
	\end{align*}
	To see that $\frac{\partial f}{\partial \alpha} <0$ we first note that $\frac{\partial^2 f}{\partial \alpha \partial z} <0$, when $z\geq 1$ and $\alpha \leq 1$, which implies that it is sufficient to prove that $	\frac{\partial f}{\partial \alpha}|_{z=1} <0$.  Evaluating this derivative at $z=1$ we have that $\frac{\partial f}{\partial \alpha}|_{z=1}= \frac{1-c}{2}$ which is clearly negative since $c>1$.
	
	Therefore, $v(\alpha,c)=\max_{z \geq 0} f(z, \alpha, c)$ is decreasing in $\alpha$ whenever $v(\alpha, c) \geq 0$.   We choose  $\underline{\alpha}(c)$ to be equal to a zero of function $ v(\alpha, c)$, whenever this zero exists on $[0,1)$, which will be the case when $v(0,c) \geq 0$ (recall that $v(1, c) <0$). Otherwise,  $v(\alpha,c) <0$ for all $\alpha \in [0,1]$. Since $v$ is decreasing in $c$, and decreasing in $\alpha$ whenever $v(\alpha, c) \geq 0$ we have that $\underline{\alpha}(c)$ decreases when $c$ increases.
\end{proof}

\end{document}